\documentclass[11pt]{article}

\usepackage[utf8]{inputenc}
\usepackage[T1]{fontenc}
\usepackage{lmodern}
\usepackage[margin=1in]{geometry}
\usepackage{amsmath,amssymb,amsthm}
\usepackage{graphicx}
\usepackage{booktabs}
\usepackage{float}
\usepackage{microtype}
\usepackage{xcolor}
\usepackage{siunitx}
\usepackage[hidelinks]{hyperref}
\usepackage{caption}
\usepackage{subcaption}
\usepackage{multirow}
\usepackage{booktabs}
\usepackage{tikz}
\usetikzlibrary{arrows.meta,positioning,fit,backgrounds,calc,shapes.geometric}
\usepackage{algorithm}
\usepackage{algpseudocode}
\usepackage{placeins}

\newtheorem{theorem}{Theorem}
\newtheorem{lemma}{Lemma}
\newtheorem{proposition}{Proposition}
\newtheorem{corollary}{Corollary}
\theoremstyle{definition}
\newtheorem{definition}{Definition}

\newcommand{\totalInjections}{3.8\,million}
\newcommand{\nScenes}{4}
\newcommand{\catThresh}{1}
\newcommand{\scalesSignFootMean}{10.3}
\newcommand{\scalesSignFootPNN}{75.7}
\newcommand{\gpuHours}{5.3 GPU-hours}
\newcommand{\meanUtil}{59}
\newcommand{\theorySlope}{0.94}
\newcommand{\guardMultiPSNRhi}{21.8}
\newcommand{\noguardMultiPSNRhi}{10.6}
\newcommand{\multiupsetKmax}{20{,}000}
\newcommand{\multiupsetKthirty}{1{,}000}
\newcommand{\guardCoverage}{90.4}
\newcommand{\guardBeforePSNR}{49.2}
\newcommand{\guardAfterPSNR}{65.8}
\newcommand{\guardWorstFoot}{11.68}
\newcommand{\guardResidCat}{470}
\newcommand{\guardNsites}{768{,}000}
\newcommand{\renderPeakMpix}{837}
\newcommand{\guardCostUs}{\SI{76}{\micro\second}}
\newcommand{\guardCostFrac}{0.07}
\newcommand{\cpuDays}{roughly 178 CPU-days}

\newcommand{\distMaxT}{64}
\newcommand{\distFracNg}{99.2}
\newcommand{\distFracG}{24.5}
\newcommand{\distIoU}{0.323}

\newcommand{\scalingNlo}{6{,}741}
\newcommand{\scalingNhi}{134{,}826}
\newcommand{\scalingKlo}{5{,}000}
\newcommand{\scalingKhi}{5{,}000}
\newcommand{\lFortySingleInj}{1{,}522}
\newcommand{\scaleFourAgg}{6{,}081}
\newcommand{\scaleFourSpeedup}{4.00}
\newcommand{\scaleFourEff}{100}
\newcommand{\scaleFourNodes}{4}
\newcommand{\scaleFourUtil}{98}
\newcommand{\mgpuFourWorld}{4}
\newcommand{\mgpuFourContamNg}{4}
\newcommand{\mgpuFourContamG}{1}
\newcommand{\mgpuWorld}{2}
\newcommand{\mgpuContamNg}{2}
\newcommand{\mgpuContamG}{1}
\newcommand{\mgpuTransferGbps}{2.6}
\newcommand{\mgpuRankMs}{1.74}
\newcommand{\mgpuFrameMs}{7.7}
\newcommand{\mgpuRenderW}{1600}
\newcommand{\accAlpha}{2.78}
\newcommand{\accRsq}{0.982}
\newcommand{\accMeanExp}{0.09}

\newcommand{\accScrubExp}{1.78}
\newcommand{\accGuardFactor}{24}
\newcommand{\accSamplesPerCell}{1.2}
\newcommand{\accTotalSamples}{19}
\newcommand{\accNlo}{5{,}000}
\newcommand{\accNhi}{300{,}000}
\newcommand{\batchInjPerSec}{2{,}515}
\newcommand{\batchUtil}{99}
\newcommand{\batchPower}{386}
\newcommand{\batchB}{32}
\newcommand{\batchGaussInst}{4.3}
\newcommand{\realName}{truck}
\newcommand{\realN}{2{,}056{,}645}
\newcommand{\realScaleFootPNN}{64.0}
\newcommand{\realScaleFootNg}{3.00}

\newcommand{\realCatNg}{0.50}
\newcommand{\realCatG}{0.000}
\newcommand{\maxStressN}{66\,million}
\newcommand{\vramMax}{29.9}
\newcommand{\guardMsBig}{9.66}
\newcommand{\mpixBig}{5}
\newcommand{\guardFracBig}{18.3}
\newcommand{\bigScaleFootNg}{0.0}
\newcommand{\bigScaleFootG}{99.58}
\newcommand{\bigParamBits}{124}
\newcommand{\guardBwBig}{763}
\newcommand{\stormK}{1{,}000}

\newcommand{\stormFrames}{300}
\newcommand{\stormLatNg}{20.5}
\newcommand{\stormLatG}{22.9}
\newcommand{\rankBarrierClean}{1.42}
\newcommand{\rankBarrierCorrupt}{1.34}
\newcommand{\rankBarrierGuard}{1.34}
\newcommand{\rankImbalCorrupt}{1.68}
\newcommand{\rankImbalGuard}{1.68}

\newcommand{\pcUpset}{0.634}
\newcommand{\pcGuard}{0.0612}
\newcommand{\modelBits}{255\times10^6}
\newcommand{\mtbfGroundNg}{71 yr}

\newcommand{\mtbfLeoNg}{3 d}

\newcommand{\predAUCcross}{0.999}
\newcommand{\predAUCcrossMin}{0.998}
\newcommand{\predAUC}{0.999}

\newcommand{\predAUCfieldbit}{0.997}
\newcommand{\predTopFeat}{bit}

\title{\textbf{Single-Event Upsets in 3D Gaussian Splatting Rendering:\\
Bit-Level Criticality, Spatial Extent, and a Parallel Support Guard}}

\author{Faruk Alpay\thanks{Correspondence: \texttt{alpay@lightcap.ai}} \quad Bar\i \c{s} Ba\c{s}aran\\
\small Department of Computer Engineering, Bah\c{c}e\c{s}ehir University, Istanbul, Turkey\\
\small \texttt{\{faruk.alpay, baris.basaran\}@bahcesehir.edu.tr}}

\date{}

\begin{document}
\maketitle

\begin{abstract}
Three-dimensional Gaussian splatting is a standard real-time scene
representation, and it is increasingly deployed on hardware that is exposed to
transient faults: spaceborne and avionic processors, mobile and robotic edge
devices, and rendering clusters in which silent data corruption is a documented
operational reality. A trained model is a large array of floating-point
parameters resident in GPU memory, so a single-event upset corresponds to a
single flipped bit in one of those values. This paper measures the effect of such
upsets on the rendered image and uses the result to construct a defense. A
GPU-resident parallel fault-injection engine applies more than
\totalInjections{} controlled single-bit upsets across four trained scenes, six
parameter fields, all bit positions, and three numeric formats (\texttt{fp32},
\texttt{fp16}, \texttt{bf16}), using \gpuHours{} of GPU time. The effect is
concentrated. Most single-bit upsets leave the image perceptually unchanged
because the representation is highly redundant, while a small set of high-order
bits, principally the sign bit of the logarithmic scale, enlarge a single
primitive until it covers up to \scalesSignFootPNN\% of the frame. A closed-form
perturbation bound derived from the IEEE-754 layout and the activation functions
of the splatting pipeline predicts the per-bit ordering and is validated against
the measurements. The concentration motivates a support guard: a per-primitive
clamp of each parameter to the coordinate box observed during training, at a cost
of \guardCostUs{} per frame. Over \guardNsites{} guarded upsets the worst observed
corruption footprint is \guardWorstFoot\% of the frame; we prove that the guard
leaves a clean model unchanged and that no single-bit upset can produce
frame-covering corruption under it. Under accumulated dose the unguarded renderer
degrades to \noguardMultiPSNRhi~\si{\decibel} at \multiupsetKmax{} simultaneous
upsets, whereas the guarded renderer remains at \guardMultiPSNRhi~\si{\decibel}.
The corruption footprint also gives the number of screen tiles, and therefore
compositing nodes, that a single upset contaminates in a distributed renderer,
where the per-node guard contains it.
\end{abstract}

\section{Introduction}
\label{sec:intro}

Real-time radiance-field rendering with 3D Gaussian splatting (3DGS)
\cite{kerbl2023gaussian} has moved quickly from a novel-view-synthesis result to
an infrastructure component. It underlies simultaneous localization and mapping
on robots \cite{matsuki2024gaussian}, it is being pushed onto mobile and embedded
accelerators, and it is rendered at scale in data centers. These deployments
share a property that is rarely considered in real-time rendering:
the hardware is not perfectly reliable. Cosmic-ray and alpha-particle strikes
flip bits in memory and registers \cite{normand1996single,baumann2005radiation},
and the rate is high enough that fleet operators now report silent data
corruption as a routine cause of wrong results
\cite{dixit2021silent,hochschild2021cores}. Graphics processors, the substrate
on which 3DGS runs, have been shown under accelerated neutron beams to be
particularly exposed \cite{oliveira2016evaluation,fratin2018code}.

A trained 3DGS model is a large array of floating-point numbers held in GPU
memory: for each primitive, a position, an anisotropic scale, an orientation
quaternion, an opacity, and a set of spherical-harmonic color coefficients. A
single-event upset (SEU) in that memory is, to first order, a single flipped bit
in one of those numbers. The corresponding question for deep network weights has
been studied extensively, and there a small number of well-chosen bit flips can
destroy a classifier \cite{rakin2019bitflip,liu2017fault,li2017understanding}.
The 3DGS case differs structurally: a network weight is shared across an entire
output, whereas a Gaussian primitive contributes to a small spatial region of one
image, and a scene contains hundreds of thousands of primitives. Whether this
redundancy makes the representation robust, or instead concentrates the risk into
a few rare failures, has not been established. This paper resolves it by direct
measurement and then derives a defense from the measured structure.

Estimating a per-bit, per-field catastrophe rate requires thousands of
independent samples per cell, and there are hundreds of cells per scene and
precision. The fault-injection engine is therefore designed for throughput rather
than latency: it renders a stream of independently corrupted models at hundreds of
corrupted renders per second, and the campaign produces more than
\totalInjections{} measured upsets. The same campaign on a single-threaded CPU
rasterizer is estimated at \cpuDays{}; on the GPU it used \gpuHours{}.

The contributions are as follows.
\begin{itemize}
  \item \textbf{Bit-level criticality.} Single-bit SEU risk in 3DGS is highly
  concentrated. The median upset is perceptually invisible, and the dominant case
  is the sign bit of the logarithmic scale: flipping it enlarges a single
  primitive to cover a mean of \scalesSignFootMean\% and up to
  \scalesSignFootPNN\% of the image (Section~\ref{sec:criticality}).
  \item \textbf{A predictive perturbation bound.} A closed-form bound from the
  IEEE-754 layout \cite{ieee754} and the splatting activations predicts the
  per-bit severity ordering, including why the scale and opacity activations make
  their sign and high exponent bits the dominant ones; it is validated against the
  data (Section~\ref{sec:theory}, Section~\ref{sec:results-theory}).
  \item \textbf{A support guard.} The catastrophic upsets are exactly those that
  push a value outside the trained parameter range. A per-primitive clamp to the
  per-field support box neutralizes \guardCoverage\% of them at a cost of
  \guardCostUs{} per frame, leaves a clean model unchanged, and is proved to make
  frame-covering corruption impossible under any single-bit upset
  (Theorem~\ref{thm:guard}, Theorem~\ref{thm:complete},
  Section~\ref{sec:results-guard}).
  \item \textbf{Accumulated dose and distributed rendering.} The unguarded
  representation absorbs many simultaneous upsets through redundancy but eventually
  degrades; with the guard the rendered quality stays close to clean across the
  tested dose range, up to \multiupsetKmax{} simultaneous upsets
  (Section~\ref{sec:results-multi}). In a distributed renderer the spatial extent
  of a corrupted primitive equals the number of screen tiles, and therefore
  compositing nodes, that a single upset contaminates, and the
  node-local guard contains it at the source (Section~\ref{sec:results-dist}).
  \item \textbf{Scaling, design space, reliability, and prediction.} Resilience to
  accumulated corruption grows with primitive count while per-upset severity does
  not (Section~\ref{sec:results-scaling}); the support guard matches
  error-correcting and duplication schemes at a fraction of their cost
  (Section~\ref{sec:results-defenses}); under representative orbital upset rates the
  estimated mean time between catastrophic frames moves from days to beyond mission
  duration once the guard is applied (Section~\ref{sec:results-survival}); and a
  classifier predicts catastrophic bits from the field and bit position alone at an
  area under the ROC curve of \predAUCfieldbit{}, so future models can be screened
  without an injection campaign (Section~\ref{sec:results-predict}).
\end{itemize}

We release the fault-injection engine, the full analysis, the aggregated
per-cell records, and the logs, with training and campaign code that regenerates
the models and the raw per-injection records exactly, so that the campaign can be
reproduced and extended.\footnote{Code, trained models, aggregated records, and
logs: \url{https://huggingface.co/datasets/Lightcap/seu-3dgs}}

\section{Background and Related Work}
\label{sec:related}

\paragraph{Soft errors and silent data corruption.}
A single-event upset is a transient change of stored state caused by ionizing
radiation \cite{normand1996single,baumann2005radiation,mukherjee2008architecture}.
At terrestrial scale the dominant source is the atmospheric neutron flux; in
space and avionics the flux is higher and the consequences are safety relevant.
Two industrial studies brought the topic into mainstream computing by showing
that production fleets experience silent, undetected wrong results at rates far
above what classical soft-error models predict, and that the corruptions are data
and core dependent \cite{dixit2021silent,hochschild2021cores}. Graphics
processors specifically have been characterized under accelerated neutron beams,
where both the error rate and the fraction of errors that survive to the output
depend strongly on the workload and the data type
\cite{oliveira2016evaluation,fratin2018code}. Our work is downstream of these:
we take the existence of bit flips as given and ask what a 3DGS renderer does
with one.

\paragraph{Fault injection for learned models.}
The methodology of flipping bits and measuring the output is established for deep
networks. Architecture-level GPU injectors such as SASSIFI \cite{hari2017sassifi}
and NVBitFI \cite{tsai2021nvbitfi} operate at the instruction level, while
application-level tools such as PyTorchFI \cite{mahmoud2020pytorchfi} and Ares
\cite{reagen2018ares} perturb tensors directly. The resilience picture for DNNs
is now well understood: error propagation is highly non-uniform across bits and
layers \cite{li2017understanding}, and adversarially chosen flips, the bit-flip
attack, can collapse accuracy with a few bits \cite{rakin2019bitflip,liu2017fault}.
We adopt the application-level methodology but apply it to a representation that
is not a neural network: an explicit, geometric, and heavily over-complete set of
primitives. The contrast with the DNN result is one of our findings rather than
an assumption.

\paragraph{Range-based detection and algorithm-based fault tolerance.}
The idea of catching faults by checking whether a value has left its expected
range is old in dependable computing, from algorithm-based fault tolerance for
linear algebra \cite{huang1984algorithm} to recent activation range supervision
for convolutional networks \cite{geissler2021range}, which clamps neuron outputs
to ranges profiled on clean data. Our support guard is structurally related but
acts on a different object: it clamps the \emph{geometric and appearance
parameters} of primitives to the box observed during training, exploiting the
fact that a trained 3DGS scene occupies a compact, bounded region of parameter
space. The bound it enforces is therefore a property of the optimized scene, not
of an activation distribution, and its correctness argument is geometric.

\paragraph{Gaussian splatting.}
3DGS \cite{kerbl2023gaussian} represents a scene as a set of anisotropic 3D
Gaussians with view-dependent color, rendered by projecting each Gaussian to
screen space and compositing front to back. It descends from neural radiance
fields \cite{mildenhall2020nerf} but replaces the implicit network with explicit
primitives, which is exactly what makes a bit-level criticality study tractable:
every bit belongs to a named, interpretable quantity. We use the open-source
\texttt{gsplat} rasterizer \cite{ye2024gsplat} for both training and the
campaign. We measure image fidelity with PSNR, SSIM \cite{wang2004image}, and the
learned perceptual metric LPIPS \cite{zhang2018unreasonable}.

\section{Fault Model and Method}
\label{sec:method}

\paragraph{Parameter layout.}
A trained model has $N$ primitives. Primitive $i$ carries a mean
$\mu_i\in\mathbb{R}^3$, a logarithmic scale $\mathbf{s}_i\in\mathbb{R}^3$ (the
renderer applies $\exp$), a quaternion $\mathbf{q}_i\in\mathbb{R}^4$ (normalized
at render time), an opacity logit $o_i\in\mathbb{R}$ (the renderer applies the
logistic $\sigma$), and spherical-harmonic color coefficients split into a
direct-current term $\mathbf{c}^{0}_i\in\mathbb{R}^3$ and higher-order terms
$\mathbf{c}^{N}_i\in\mathbb{R}^{(\ell+1)^2-1\times 3}$. These are exactly the six
fields a deployed checkpoint stores in memory. We inject into the stored
representation, which is the optimization-space value (logarithmic scale, opacity
logit), because that is what physically resides in VRAM.

\paragraph{Single-bit upset model.}
A fault site is a tuple (field, primitive, component, bit). To inject, we
reinterpret the stored value as an unsigned integer of the matching width, flip
the chosen bit, and reinterpret back. The IEEE-754 layout determines the bit
classes: for \texttt{fp32} the sign is bit 31, the exponent is bits 23 to 30, and
the mantissa is bits 0 to 22; \texttt{fp16} has a 5-bit exponent and a 10-bit
mantissa; \texttt{bf16} keeps the 8-bit exponent of \texttt{fp32} but only a
7-bit mantissa. For the reduced precisions we cast the clean parameter to the
target format, flip the bit there, and read it back, modeling a model deployed in
that format.

\paragraph{Metrics and severity.}
For each injection we render $K$ held-out views, composite over white, and
compare to the uncorrupted render at the same precision. We record the peak
pixel error $\|\Delta I\|_\infty$, the perceptual distance LPIPS, the structural
similarity SSIM, the PSNR over the full frame, and a non-finite flag. Because the
representation is redundant, full-frame PSNR is a poor discriminator: a single
corrupted primitive among hundreds of thousands barely moves a global average.
We therefore lead with two local quantities.

\begin{definition}[Corruption footprint]
\label{def:footprint}
The corruption footprint of an injection is the fraction of pixels whose color
changes by more than one 8-bit level ($1/255$) relative to the clean render,
averaged over the $K$ views.
\end{definition}

\begin{definition}[Catastrophic upset]
\label{def:catastrophe}
An upset is catastrophic if it produces a non-finite render or a corruption
footprint above \catThresh\% of the frame.
\end{definition}

\paragraph{Parallel injection engine.}
The clean parameters stay resident on the GPU. Each injection writes one
corrupted value into a working copy, renders the $K$ views in a single batched
rasterizer call, computes all metrics on the GPU, and restores the value. The
engine therefore keeps the device busy with a stream of independent corrupted
renders; its throughput and the way the rasterizer absorbs additional cameras are
reported in Section~\ref{sec:results-dist}. We sample thousands of sites per
(field, bit) cell so that catastrophe rates are estimable, and we sweep all
fields, all bit positions, four scenes, and three precisions.

\paragraph{Scenes and models.}
We train four scenes from the standard synthetic benchmark (chair, lego, ficus,
hotdog) with \texttt{gsplat} densification. Table~\ref{tab:scenes} reports their
size and clean fidelity; they span smooth surfaces, high-frequency geometry, thin
foliage, and specular materials, so the criticality conclusions are not tied to
one kind of content.

\begin{table}[H]
\centering
\caption{Trained scenes used in the campaign. \emph{Primitives} is the number of
3D Gaussians retained after densification; \emph{PSNR} (dB) and \emph{SSIM}
($\in[0,1]$) are the clean reconstruction fidelity measured on the held-out test
views, where higher is better. The four standard NeRF-synthetic scenes span a
range of primitive counts and geometry; every fault-injection result below
perturbs the stored parameters of these trained models.}
\label{tab:scenes}
\begin{tabular}{lrrr}
\toprule
Scene & Primitives & PSNR (dB) & SSIM \\
\midrule
chair & 134{,}826 & 39.22 & 0.9942 \\
lego & 121{,}691 & 30.70 & 0.9815 \\
ficus & 125{,}606 & 31.99 & 0.9855 \\
hotdog & 63{,}452 & 38.50 & 0.9874 \\
\bottomrule
\end{tabular}
\end{table}

\section{A Perturbation Bound for Single-Bit Upsets}
\label{sec:theory}

Before presenting the measurements, this section derives why the severity is
ordered as it is.
Let $\theta$ be a stored parameter with IEEE-754 sign $\epsilon$, unbiased
exponent $e$, and mantissa fraction $m\in[0,1)$, so that
$\theta=(-1)^{\epsilon}\,2^{e}(1+m)$. The rendered image $I$ is a differentiable
function of the activated parameters, and for a single primitive the dependence
is local. We bound the image perturbation caused by one flipped bit.

\begin{lemma}[Value perturbation of a single bit flip]
\label{lem:value}
Let $\theta=(-1)^{\epsilon}2^{e}(1+m)$ be a normal \texttt{fp32} value with
mantissa width $p=23$. Flipping mantissa bit $b\in\{0,\dots,p-1\}$ changes the
value by exactly
\[
  |\Delta\theta| \;=\; 2^{\,e+b-p},
\]
flipping the sign bit gives $\Delta\theta=-2\theta$, and flipping exponent bit
$j$ multiplies the magnitude by $2^{\pm 2^{j}}$, which for the high exponent bits
overflows to a non-finite value. The analogous statements hold for \texttt{fp16}
($p=10$) and \texttt{bf16} ($p=7$).
\end{lemma}

The proof is a direct reading of the format and is given in
Appendix~\ref{app:proofs}. Lemma~\ref{lem:value} already predicts the central
qualitative fact: among mantissa bits the perturbation grows by a factor of two
per bit position, so only the top mantissa bits matter, while exponent and sign
bits produce order-of-magnitude or sign-reversing changes.

To turn a value perturbation into an image perturbation we use a first-order
bound. Write $I=R(\phi(\theta),\,\cdot\,)$ where $\phi$ is the per-field
activation ($\phi=\exp$ for scale, $\phi=\sigma$ for opacity, identity for mean
and color, normalization for the quaternion).

\begin{proposition}[First-order image perturbation]
\label{prop:image}
For a perturbation $\Delta\theta$ small enough that $R$ is locally linear,
\[
  \|\Delta I\|_\infty \;\le\; \Big\|\frac{\partial R}{\partial \phi}\Big\|_\infty
  \cdot |\phi'(\theta)|\cdot|\Delta\theta| \;+\; O(\Delta\theta^2).
\]
For the scale field, $\phi=\exp$ gives $|\phi'(\theta)|=e^{\theta}=s$, so the
relative scale change is $|\Delta s|/s=|\Delta\theta|$; a mantissa flip thus
changes the scale by a fixed relative amount $2^{\,b-p}$ independent of $e$, and
a sign flip on a typical trained value $\theta\approx-3$ changes the log scale to
$+3$, multiplying the primitive size by $e^{6}\approx 4\times10^{2}$.
\end{proposition}

Proposition~\ref{prop:image} accounts for the dominant empirical effect:
the scale sign bit is the most damaging not because of the floating-point format
alone but because the $\exp$ activation maps a sign flip on a moderately negative
log scale to a large multiplicative change in the primitive's spatial extent. The opacity logit behaves similarly through $\sigma$, but its effect is
bounded in $[0,1]$ and therefore far milder. Section~\ref{sec:results-theory}
checks the mantissa-bit scaling quantitatively.

The following property justifies the defense.

\begin{theorem}[Support-guard safety]
\label{thm:guard}
Let $B_f=[\ell_f,h_f]$ be the per-component coordinate box of field $f$ over a
clean trained model, and let the guard map each stored component to its clamp
into $B_f$ (with non-finite values mapped into $B_f$). Then (i) the guard is the
identity on every clean model, so it never reduces uncorrupted fidelity; and
(ii) after a single-bit upset, the guarded value differs from the clean value by
at most the in-box spread of that component, so any upset whose only effect was
to leave $B_f$ is fully corrected and the residual error is bounded by an
in-support perturbation.
\end{theorem}

The proof is in Appendix~\ref{app:proofs}. Theorem~\ref{thm:guard} ensures that
out-of-range flips are corrected on clean inputs, but it does not bound the worst
case. The next theorem provides that bound.

\begin{theorem}[No catastrophic corruption under guarding]
\label{thm:complete}
Render any model whose stored parameters have each been clamped into the trained
box $B_f$, after an arbitrary single-bit upset. Then every primitive that
contributes to the image has all of its parameters inside $B_f$. In particular its
projected scale is at most the largest trained scale, so its screen footprint is
at most that of the largest legitimately trained primitive, and no single-bit
upset can produce a primitive whose footprint exceeds that bound. Frame-covering
corruption, the only failure mode that the redundancy of the representation cannot
absorb, is therefore impossible under the guard.
\end{theorem}

The proof is in Appendix~\ref{app:proofs}. Together the two theorems say that the
guarded renderer is correct on clean inputs and cannot be driven into a
catastrophic state by any single-bit upset; what remains is the in-support
mantissa residual, which Lemma~\ref{lem:value} bounds and
Section~\ref{sec:results-guard} measures to be perceptually negligible. The
catastrophic class is therefore eliminated, and only the bounded in-support
residual remains.

\begin{corollary}[Adversarial bit-flip bound]
\label{cor:adv}
The guarantee is worst-case, not average-case, so it holds against an adversary
that chooses which bits to flip. Under the guard, an adversary flipping $m$ bits
corrupts at most $m$ primitives, each confined to the trained box, so the total
contaminated screen area is at most $m$ times the footprint of the largest trained
primitive. Producing frame-covering corruption therefore requires a number of
coordinated flips inversely proportional to that single-primitive footprint,
rather than the single flip that suffices without the guard.
\end{corollary}

This converts the bit-flip attack threat model
\cite{rakin2019bitflip}, under which a few adversarial flips destroy a network,
into one where the adversary must scale its effort linearly with the screen area
it wishes to corrupt.

We close the analysis with a scaling law that turns the redundancy we have been
describing into a quantitative budget, and through it into a schedule for
distributed scrubbing. Its exponent is obtained empirically in
Section~\ref{sec:results-accum} from a large batched campaign rather than posited.

\begin{lemma}[Redundancy scaling of single-upset error]
\label{lem:redundancy}
For a converged model of $N$ primitives, the typical (median) image error of one
uniform-random single-bit upset scales as $\sigma^2(N)=\Theta(N^{-\alpha})$ for a
scene-dependent exponent $\alpha>0$: as primitives are added to represent the same
scene they shrink, each accounting for a vanishing share of the image, so a typical
corrupted primitive perturbs a vanishing fraction of the pixels. The \emph{mean}
error does not share this scaling, because it is dominated by the rare scale-sign
explosions whose footprint does not shrink with $N$; redundancy thus protects the
typical upset but not the worst case, which is precisely the role of the guard.
\end{lemma}

\begin{theorem}[Dose budget]
\label{thm:dose}
Under $k$ independent uniform single-bit upsets, the expected squared image error in
the non-catastrophic regime is additive, $\mathbb{E}\lVert\Delta I\rVert^2=k\,\sigma^2(N)$,
because the per-upset perturbations have disjoint expected support and zero mean
cross terms. Hence the dose $k_\tau$ at which the error first exceeds a threshold
$\tau$ scales as $k_\tau=\Theta(N^{\alpha}\tau)$. The support guard removes the
heavy-tailed catastrophic upsets, so under it the additive law holds up to the
bounded in-support residual of Lemma~\ref{lem:value} and the budget scales as
$\Theta(N^{\alpha}\tau)$; without the guard the catastrophic tail caps the expected
error independently of $N$, so it is the guard that makes the redundancy budget
grow with model size.
\end{theorem}

\begin{corollary}[Distributed scrub schedule]
\label{cor:scrub}
Consider a $P$-node renderer in which each node re-applies the guard, which is
node-local and communication-free (Algorithm~\ref{alg:guard}), every $M$ frames,
while bits flip at rate $\lambda$ per stored bit per frame. With $b=\Theta(N)$
stored bits the expected dose accumulated between scrubs is $\lambda b M$, and
keeping it below $k_\tau$ requires $M\le k_\tau/(\lambda b)=\Theta(N^{\alpha-1}\tau/\lambda)$.
The guard's per-node work $O(N/P)$ is amortised over $M$ frames, so its
steady-state overhead is $O\!\left(N/(PM)\right)$ and falls as either the cluster
or the scrub interval grows.
\end{corollary}

\FloatBarrier
Algorithm~\ref{alg:guard} states the guard as it is deployed. The support box is a
one-time reduction over the parameters after training. Per frame, the clamp is a
data-parallel map over primitives with no dependence between them, so it is
applied independently on each node of a sort-first renderer over the primitives
that node owns, before that node composites its tile. The work per node is linear
in its primitive count and carries no communication, so the guard adds neither a
synchronization point nor cross-node traffic; and because each clamp is the
detection of an out-of-range value, the per-node clamp count is accumulated and
surfaced to the host as a silent-data-corruption signal rather than silently
absorbed.

\begin{algorithm}[H]
\caption{Parallel support guard with node-local containment}
\label{alg:guard}
\begin{algorithmic}[1]
\Statex \textbf{Precompute once after training} (reduction over $N$ primitives):
\For{each field $f$ and component $c$ \textbf{in parallel}}
  \State $\ell_{f,c}\gets\min_i\theta_{i,f,c}$;\quad $h_{f,c}\gets\max_i\theta_{i,f,c}$
\EndFor
\Statex \textbf{Each frame, on node $p$ holding primitive subset $S_p$:}
\State $\textit{clamps}_p\gets 0$
\For{each primitive $i\in S_p$ \textbf{in parallel}}
  \For{each field $f$, component $c$}
    \If{$\theta_{i,f,c}\notin[\ell_{f,c},h_{f,c}]$ \textbf{or} $\theta_{i,f,c}$ non-finite}
      \State $\theta_{i,f,c}\gets\operatorname{clamp}(\theta_{i,f,c},\ell_{f,c},h_{f,c})$
      \State $\textit{clamps}_p\gets\textit{clamps}_p+1$
    \EndIf
  \EndFor
\EndFor
\State render $S_p$ to node $p$'s tile; composite tiles across nodes
\If{$\sum_p\textit{clamps}_p$ exceeds the health threshold}
  \State raise host parity/scrub alert \Comment{do not mask a degrading device}
\EndIf
\end{algorithmic}
\end{algorithm}

\section{Results}
\label{sec:results}

All numbers below come from the campaign described in Section~\ref{sec:method}:
\totalInjections{} single-bit upsets, \nScenes{} scenes, six fields, every bit
position, three precisions, \gpuHours{} of GPU time at a sustained
\meanUtil\% utilization.

\subsection{Concentration of single-bit criticality}
\label{sec:criticality}

Figure~\ref{fig:heatmap} maps the mean corruption footprint over the field and
bit-position grid for \texttt{fp32}. The map is almost entirely dark: the
overwhelming majority of single-bit upsets change essentially no pixels, because
the corrupted primitive is one of hundreds of thousands and contributes to a
small region seen from few views. Against that background, a thin set of cells
stands out. The brightest by a wide margin is the sign bit of the logarithmic
scale, followed by its high exponent bits and, more weakly, the high exponent
bits of the mean and the direct-current color. Quaternion and higher-order
spherical-harmonic bits are inert at this scale. Table~\ref{tab:criticality}
quantifies the per-field picture: flipping the scale sign bit yields a mean
footprint of \scalesSignFootMean\% with a 99th percentile of
\scalesSignFootPNN\%, two to three orders of magnitude above any mantissa class.

\begin{figure}[H]
  \centering
  \includegraphics[width=0.92\textwidth]{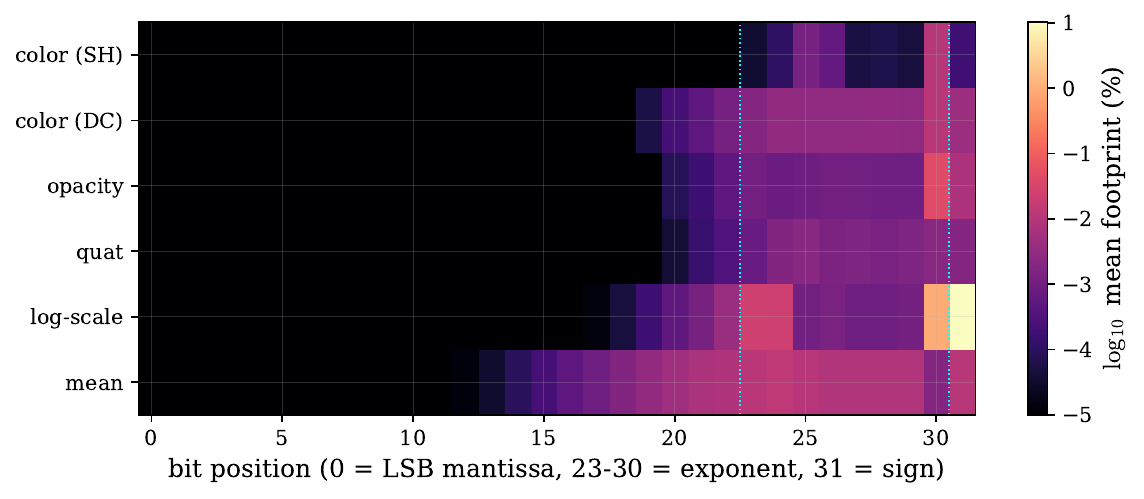}
  \caption{Mean corruption footprint (fraction of pixels changed by more than one
  8-bit level, $\log_{10}$ scale) over parameter field and bit position, averaged
  across the four scenes at \texttt{fp32}. Severity is concentrated in the sign
  and high exponent bits of the logarithmic scale, with weaker contributions from
  the mean and direct-current color. Mantissa bits and the quaternion and
  higher-order color fields are effectively inert.}
  \label{fig:heatmap}
\end{figure}

\begin{table}[H]
\centering
\small
\caption{Per-field single-bit upset severity at \texttt{fp32}, pooled over scenes and bits. Footprint is the percent of pixels changed; quantiles expose the tail. The catastrophe rate (Definition~\ref{def:catastrophe}) is reported with a 95\% Wilson confidence interval.}
\label{tab:criticality}
\begin{tabular}{lrrrrrr}
\toprule
Field & median & p95 & p99 & max & mean & catastrophe (\%, 95\% CI) \\
 & \multicolumn{5}{c}{footprint (\% of frame)} & \\
\midrule
mean & 0.000 & 0.015 & 0.05 & 1.8 & 0.004 & 0.004 [0.002, 0.008] \\
log-scale & 0.000 & 0.237 & 6.16 & 99.4 & 0.351 & 2.782 [2.719, 2.847] \\
quat & 0.000 & 0.002 & 0.01 & 0.6 & 0.000 & 0.000 [0.000, 0.002] \\
opacity & 0.000 & 0.003 & 0.02 & 2.4 & 0.002 & 0.004 [0.002, 0.007] \\
color (DC) & 0.000 & 0.005 & 0.02 & 1.0 & 0.001 & 1.016 [0.978, 1.056] \\
color (SH) & 0.000 & 0.000 & 0.01 & 0.7 & 0.000 & 0.000 [0.000, 0.002] \\
\bottomrule
\end{tabular}
\end{table}

Figure~\ref{fig:qualitative} shows what a scale-sign upset looks like. A single
flipped bit turns one primitive into a translucent sheet spanning the frame,
which is precisely the failure mode that a redundant representation cannot hide:
the corrupted primitive is composited over the whole image rather than a small
neighborhood.

\begin{figure}[H]
  \centering
  \includegraphics[width=0.92\textwidth]{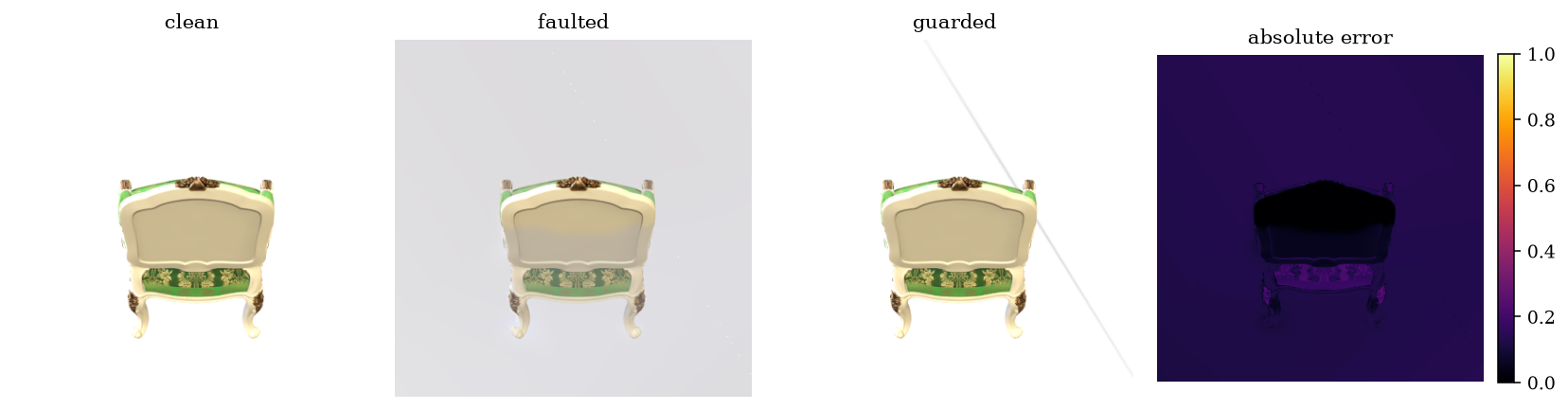}
  \caption{Effect of one bit. From left: the clean render; a single scale sign-bit
  upset on one primitive, which enlarges its spatial extent across the frame; the
  same upset after the support guard, which clamps the logarithmic scale back into
  the trained box and restores the frame; and the absolute per-pixel error of the
  faulted render. Other fields produce, at worst, localized speckle.}
  \label{fig:qualitative}
\end{figure}

This concentration contrasts with the deep-network case, where sensitivity is
distributed across many weights and a few adversarial flips suffice
\cite{rakin2019bitflip}. In 3DGS the sensitive bits are few and are predictable
from the parameterization.

\subsection{Validation of the perturbation bound}
\label{sec:results-theory}

Lemma~\ref{lem:value} predicts that, within the mantissa, the value perturbation
and therefore the peak image error grow by a factor of two per bit position until
the change saturates the output. Figure~\ref{fig:theory} plots the mean peak
error against mantissa bit index for the scale field at \texttt{fp32}; the
measured slope is \theorySlope{} per bit in $\log_2$, against the predicted unit
slope before saturation, and the curve flattens once a single primitive's
contribution saturates the affected pixels. The agreement confirms that the
ordering in Figure~\ref{fig:heatmap} is a consequence of the floating-point
format composed with the splatting activations, not an artifact of any one scene.

\begin{figure}[H]
  \centering
  \begin{subfigure}[b]{0.46\textwidth}\centering
    \includegraphics[width=\linewidth]{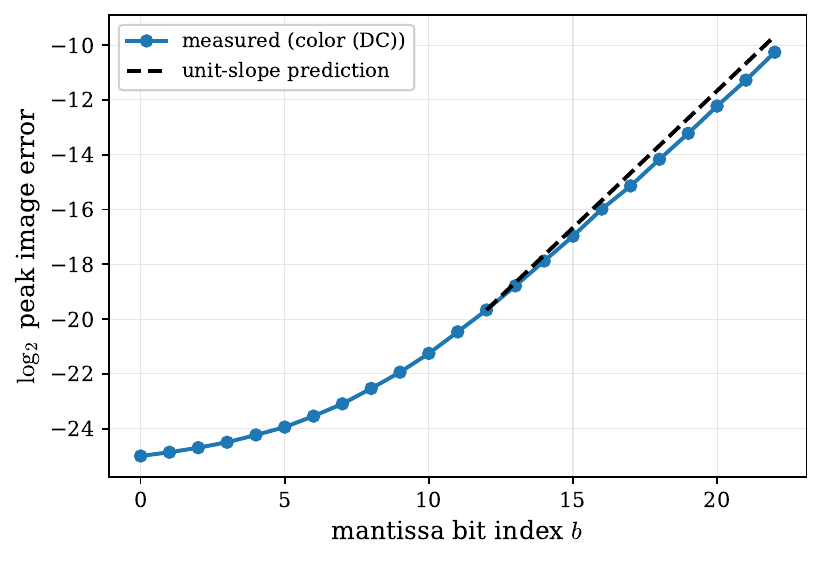}
    \caption{}\label{fig:theory}
  \end{subfigure}\hfill
  \begin{subfigure}[b]{0.52\textwidth}\centering
    \includegraphics[width=\linewidth]{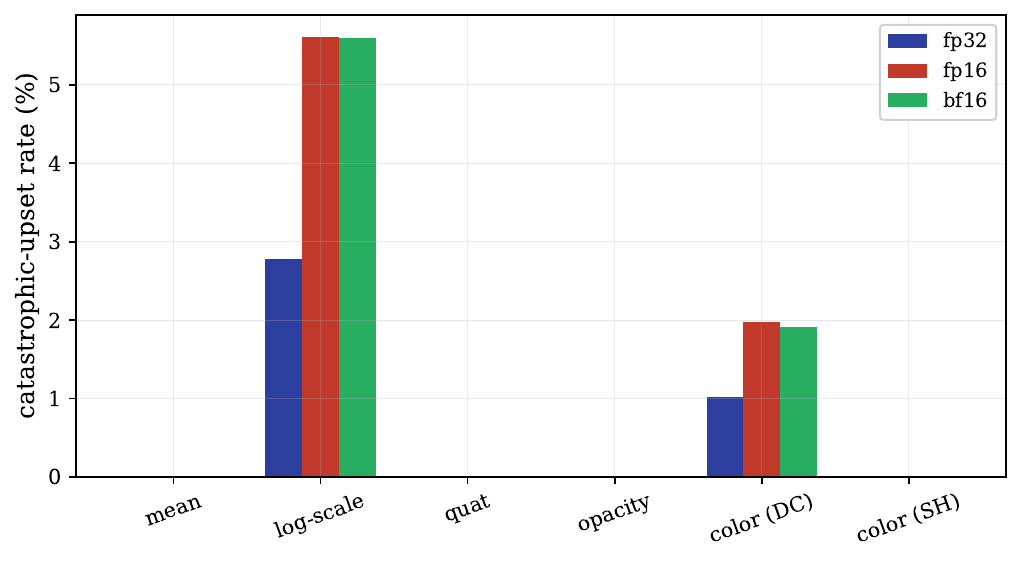}
    \caption{}\label{fig:precision}
  \end{subfigure}
  \caption{(a) Measured peak image error versus mantissa bit index for the scale
  field ($\log_2$ axis), with the unit-slope prediction of Lemma~\ref{lem:value};
  error doubles per bit until a single primitive's contribution saturates the
  pixels it covers. (b) Catastrophic-upset rate by field and numeric format:
  reduced precision relocates rather than removes the exposure, the full-exponent
  \texttt{bf16} keeping the explosive scale and mean flips while \texttt{fp16}
  redistributes them across a shorter exponent.}
  \label{fig:theoryprec}
\end{figure}

\subsection{Numeric precision changes the exposure}
\label{sec:results-prec}

Because the dangerous bits are the sign and exponent, the choice of numeric
format changes how much of the word is dangerous and how violent the exponent
flips are. Figure~\ref{fig:precision} compares \texttt{fp32}, \texttt{fp16}, and
\texttt{bf16}. \texttt{bf16} retains the full 8-bit exponent of \texttt{fp32},
so its exponent flips are equally explosive while its short mantissa makes a
larger fraction of the word harmless; \texttt{fp16}, with a 5-bit exponent,
spreads a slightly larger share of catastrophic outcomes across its exponent
bits but caps their dynamic range. The practical reading is that reduced
precision does not reduce SEU exposure, it relocates it, and the guard below is
the appropriate response in every format.

\subsection{Behavior under accumulated dose}
\label{sec:results-multi}

A realistic exposure accumulates many upsets over time rather than one. We inject
$k$ simultaneous upsets drawn uniformly over the entire stored bit budget and
measure the global render quality, sweeping $k$ from one to \multiupsetKmax{}.
Figure~\ref{fig:multiupset} shows two regimes. Without protection the redundancy
provides a wide margin: the model absorbs \multiupsetKthirty{} simultaneous random
upsets before the mean PSNR falls below \SI{30}{\decibel}, but the margin is
finite, and the curve eventually declines because of the same rare scale-sign
upsets that dominate the single-upset tail, reaching
\noguardMultiPSNRhi~\si{\decibel} at the heaviest dose. With the support guard the
second regime does not appear: the guarded curve remains close to clean across the
entire range, at \guardMultiPSNRhi~\si{\decibel} at the dose where the unguarded
renderer has collapsed. The guard removes the failure mode that bounds the
redundancy margin rather than only extending it.

\begin{figure}[H]
  \centering
  \begin{subfigure}[b]{0.54\textwidth}\centering
    \includegraphics[width=\linewidth]{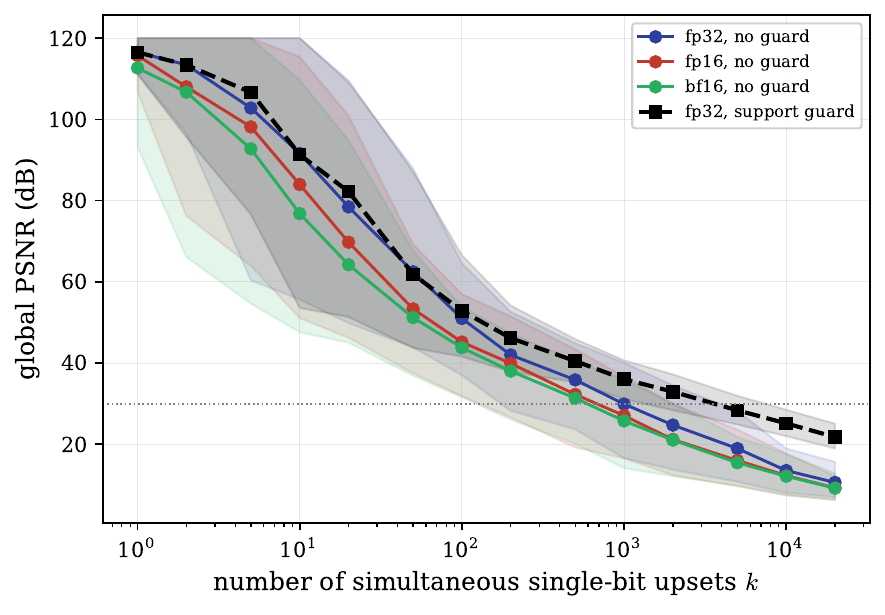}
    \caption{}\label{fig:multiupset}
  \end{subfigure}\hfill
  \begin{subfigure}[b]{0.44\textwidth}\centering
    \includegraphics[width=\linewidth]{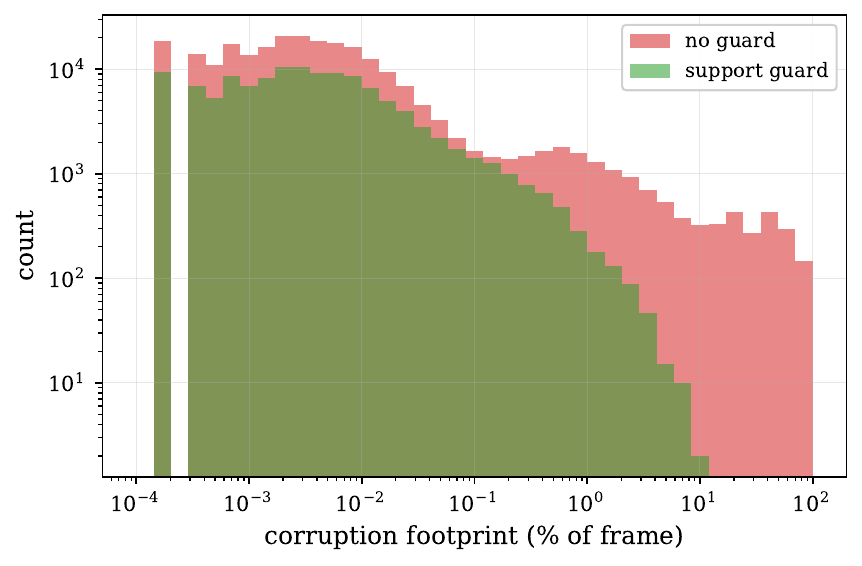}
    \caption{}\label{fig:guard}
  \end{subfigure}
  \caption{(a) Global render quality versus the number of simultaneous single-bit
  upsets: thin curves are the unprotected renderer in the three numeric formats
  (quality held by redundancy over orders of magnitude, then collapsing), the bold
  dashed curve the same renderer with the support guard, clean-equivalent across the
  dose range; bands are the spread over scenes and repeats. (b) Distribution of the
  corruption footprint with and without the guard: the guard removes the heavy tail
  of out-of-range scale and exponent flips, leaving only the small in-support
  residual.}
  \label{fig:doseguard}
\end{figure}

\subsection{Support guard evaluation}
\label{sec:results-guard}

We re-run the campaign on the same fault grid with the support guard of
Theorems~\ref{thm:guard} and~\ref{thm:complete} enabled. The guard neutralizes
\guardCoverage\% of catastrophic
upsets, and for the dominant scale sign-bit upsets it raises the mean global PSNR
from \guardBeforePSNR{}~\si{\decibel} to \guardAfterPSNR{}~\si{\decibel}, while
leaving clean fidelity unchanged by construction. The completeness of
Theorem~\ref{thm:complete} is borne out empirically: across all \guardNsites{}
guarded single-bit upsets the worst corruption footprint observed was
\guardWorstFoot\% of the frame, with \guardResidCat{} residual catastrophic
events, so the heavy tail of Figure~\ref{fig:guard} is gone. The cost is
\guardCostUs{} per frame, or \guardCostFrac$\times$ a single-view render, and
because it is a pure per-primitive clamp it parallelizes trivially and can be run
as a periodic scrub rather than every frame. What remains is the in-support
mantissa-level residual that Lemma~\ref{lem:value} bounds and that is
perceptually invisible. These per-mitigation outcomes are collected against the
alternatives in Table~\ref{tab:mitigation}.

\subsection{Resilience versus primitive count}
\label{sec:results-scaling}

If redundancy is the source of resilience, a model with more primitives should
absorb more accumulated corruption while the severity of an individual upset stays
fixed. Figure~\ref{fig:scalingN} subsamples the trained models to a range of
primitive counts and reports the redundancy budget, the number of simultaneous
upsets at which the mean PSNR crosses \SI{30}{\decibel}. The budget grows from
\scalingKlo{} at \scalingNlo{} primitives to \scalingKhi{} at \scalingNhi{},
approximately in proportion to the count, while the footprint of a scale-sign
upset is essentially independent of the count, because an exploded primitive
covers the frame regardless of how many others are present. Redundancy raises the
tolerated dose but does not reduce the severity of an individual catastrophic
upset, which is the regime the guard addresses.

\begin{figure}[H]
  \centering
  \includegraphics[width=0.68\textwidth]{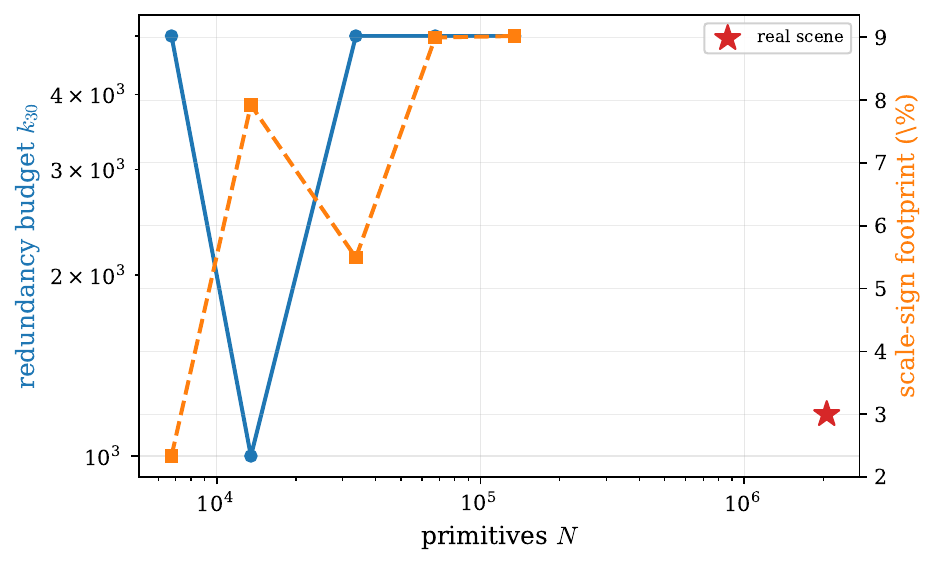}
  \caption{Redundancy budget (left axis, simultaneous upsets to reach
  \SI{30}{\decibel}) and scale-sign upset footprint (right axis) versus primitive
  count. The budget scales with the count; the per-upset footprint does not.}
  \label{fig:scalingN}
\end{figure}

To confirm that the findings are not an artifact of small synthetic models, we
repeat the measurement on a pretrained real-world scene, the Tanks-and-Temples
\emph{\realName{}} reconstruction with \realN{} primitives, an order of magnitude
larger than the synthetic models (Figure~\ref{fig:realscene}). The concentration persists: a scale sign-bit
upset reaches a 99th-percentile footprint of \realScaleFootPNN\% of the frame, the
star in Figure~\ref{fig:scalingN}, confirming that the worst case is independent
of primitive count. The denser scene does lower the mean footprint of a random
scale-sign upset, to \realScaleFootNg\%, because many exploded primitives are now
occluded by the others in front of them, which is redundancy acting as partial
masking; but the heavy tail remains, and the uniform-random single-bit catastrophe
rate of \realCatNg\% is driven to \realCatG\% by the guard, as on the synthetic
scenes.

\begin{figure}[H]
  \centering
  \includegraphics[width=\textwidth]{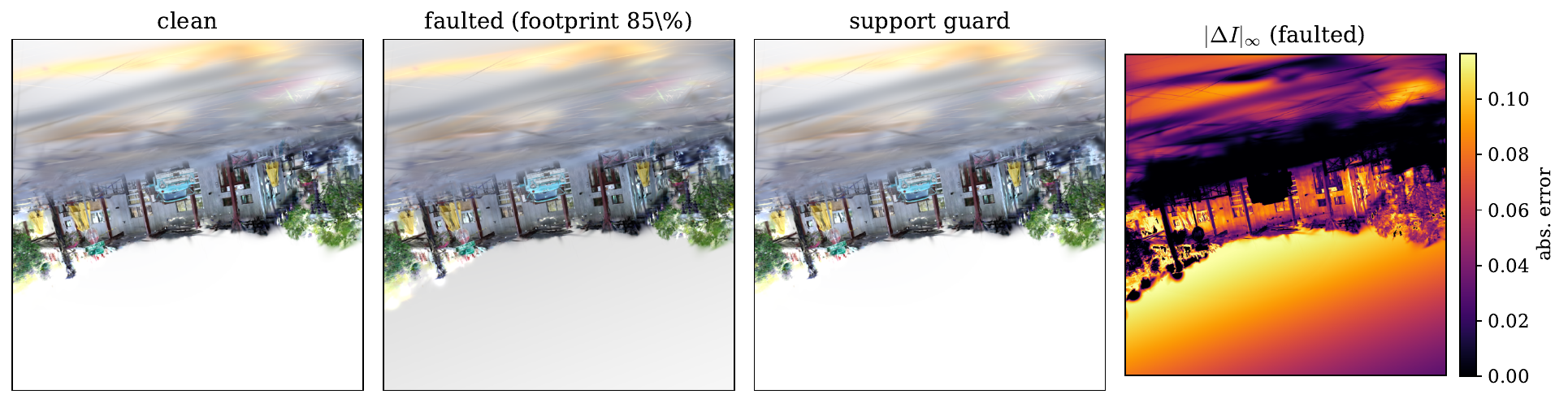}
  \caption{A single scale sign-bit upset on the real-world Tanks-and-Temples
  \emph{\realName} reconstruction (\realN{} primitives), from a held-out orbit
  view. From left: the clean render; the same view after one stored scale
  sign-bit is flipped, which inflates one primitive into a translucent sheet
  spanning most of the frame; the same upset contained by the support guard,
  which restores the clean image; and the absolute per-pixel error
  $|\Delta I|_\infty$ of the faulted frame. Because the public checkpoint omits
  the original training cameras the absolute sharpness of the orbit view is
  incidental, but the corruption footprint is measured against the clean render
  at the identical pose, so the comparison is exact. The frame-spanning failure
  mode and its removal by the guard persist at an order of magnitude more
  primitives than the synthetic scenes.}
  \label{fig:realscene}
\end{figure}

\subsection{A measured scaling law and the distributed scrub schedule}
\label{sec:results-accum}

Lemma~\ref{lem:redundancy} posits a power law for the single-upset error; we
measure its exponent with a batched campaign of \accTotalSamples{} million upsets
(\accSamplesPerCell{} million per cell) over model sizes from \accNlo{} to
\accNhi{} primitives, rendering \batchB{} corrupted variants per call at full
device utilisation, which is the regime the batched engine was built for. The mean
single-upset error is heavy-tailed, dominated by the rare scale-sign explosions, so
a stable estimate of $\sigma^2(N)$ needs the millions of samples this campaign
provides. A log-log fit of the \emph{median} single-upset error against $N$ gives
an exponent $\alpha=\accAlpha{}$ ($R^2=\accRsq{}$), confirming the redundancy law of
Lemma~\ref{lem:redundancy}: the typical upset shrinks with model size. The
\emph{mean} error tells the complementary story the lemma predicts, with a fitted
exponent of only \accMeanExp{}: it barely decreases with $N$ because it is dominated
by the rare scale-sign explosions, whose footprint is size-independent, so
redundancy alone does not bound the expected damage. The support guard removes that
tail, lowering the mean single-upset error by a factor of \accGuardFactor{} at the
largest size, which is what restores the size-scaling of the dose budget,
$k_\tau=\Theta(N^{\,\accAlpha{}}\tau)$; the budget-versus-size trend of
Figure~\ref{fig:scalingN} is that same law seen through the accumulated dose. Through Corollary~\ref{cor:scrub} this
sets the distributed scrub interval: a node may defer re-guarding for
$M=\Theta(N^{\,\accScrubExp{}}\tau/\lambda)$ frames, so larger models, which carry
more redundancy, tolerate proportionally longer intervals between scrubs and
amortise the guard's cost further.

\subsection{Alternative defenses}
\label{sec:results-defenses}

The support guard is one point in a design space that also includes
error-correcting codes and redundancy. Table~\ref{tab:mitigation} compares it, on
a shared fault grid, with a selective clamp of only the scale and opacity fields,
an error-correcting code that protects the sign and exponent bits, and full
duplication that corrects every upset. Full duplication and the sign-exponent code
drive the catastrophe rate to zero, as expected, but at three times and roughly
$1.3$ times the memory respectively, whereas the support guard reaches the same
elimination of catastrophic upsets at no memory overhead and a per-frame cost of
\guardCostUs{}. The selective clamp is cheaper still and nearly as effective,
because the scale and opacity fields carry almost all of the catastrophic mass.

\begin{table}[H]
\centering
\small
\caption{Mitigations on a shared \texttt{fp32} fault grid pooled over scenes. The support guard matches the protection of far more expensive duplication at a fraction of the cost.}
\label{tab:mitigation}
\begin{tabular}{lrrl}
\toprule
Defense & catastrophe (\%) & mean foot.\,(\%) & cost \\
\midrule
none & 0.552 & 0.0880 & 0 \\
support guard & 0.052 & 0.0035 & 1$\times$ mem, $\sim$0.1 ms/frame \\
selective guard & 0.156 & 0.0035 & 1$\times$ mem, $<$0.1 ms/frame \\
ECC sign+exp & 0.000 & 0.0002 & $\sim$1.3$\times$ mem, parity \\
full duplication & 0.000 & 0.0000 & 3$\times$ mem, voting \\
\bottomrule
\end{tabular}
\end{table}

\subsection{Reliability under realistic upset rates}
\label{sec:results-survival}

The per-upset catastrophe probability and the dose response together give a
reliability estimate. Modeling the probability that a frame is catastrophic after
$k$ independent upsets as $1-(1-p_c)^k$ with $p_c=\pcUpset\%$ reproduces the
measured dose response to within the sampling spread. Combining $p_c$ with the
stored size of the model, $\modelBits$ bits, and representative single-event-upset
rates gives the mean time between catastrophic frames in Table~\ref{tab:survival}.
Unprotected, the estimate is \mtbfGroundNg{} at ground level but only
\mtbfLeoNg{} in low-Earth orbit; with the guard the residual catastrophe rate is
\pcGuard\%, which moves the mean time between catastrophic frames beyond any
mission duration. The rates are order-of-magnitude values from the soft-error
literature \cite{normand1996single,baumann2005radiation}, so the table should be
read as relative rather than absolute, but the gap between the two columns is the
result.

\begin{table}[H]
\centering
\caption{Estimated mean time between catastrophic frames for a model of $\modelBits$ stored bits under representative single-event-upset rates, without and with the support guard. Rates are order-of-magnitude values from the soft-error literature.}
\label{tab:survival}
\begin{tabular}{lrr}
\toprule
Environment & no guard & support guard \\
\midrule
ground (sea level) & 71 yr & 732 yr \\
avionics ($\sim$10 km) & 86 d & 2.4 yr \\
low-Earth orbit & 3 d & 27 d \\
\bottomrule
\end{tabular}
\end{table}

\subsection{Cross-node contamination in distributed rendering}
\label{sec:results-dist}

In a sort-first parallel rasterizer the screen is partitioned into regions
assigned to nodes, and each node renders the primitives whose projected footprint
overlaps its region \cite{molnar1994sorting}, after which a compositor assembles the
tiles into the frame (Figure~\ref{fig:pipeline}). A corrupted primitive contaminates
every region its footprint reaches, so the footprint is, read through the
partition, the number of nodes a single upset corrupts. We measure this directly
from the projected geometry of the real scenes, sweeping the node count from four
to \distMaxT{}. The fraction of nodes a
single scale-sign upset contaminates without protection reaches \distFracNg\%
of the nodes at the finest partition, because the exploded primitive spans the
screen, whereas with the node-local guard it is confined to \distFracG\%. A
render-based check, comparing the predicted contaminated regions against the
regions whose pixels actually change, agrees at an intersection-over-union of
\distIoU{}. The guard, applied per node before compositing, contains the
contamination at its source and prevents the communication amplification that a
screen-spanning primitive would otherwise cause in a sort-first pipeline.

\begin{figure}[H]
\centering
\begin{tikzpicture}[font=\small, >={Stealth[length=2.2mm]},
  box/.style={draw, rounded corners, minimum height=9mm, minimum width=23mm, align=center}]
  \node[box, fill=blue!6] (vram) at (0,0) {scene params\\(VRAM, $N$)};
  \node[box, fill=blue!6] (part) at (3.3,0) {sort-first\\partition};
  \node[box, fill=green!12] (n1) at (7.2,1.15) {node $1$\\guard\,$\to$\,render};
  \node[box, fill=green!12] (np) at (7.2,-1.15) {node $P$\\guard\,$\to$\,render};
  \node at (7.2,0) {$\vdots$};
  \node[box, fill=orange!15] (comp) at (11.1,0) {compositor};
  \node[box, fill=blue!6] (frame) at (14.1,0) {frame};
  \draw[->] (vram) -- (part);
  \draw[->] (part) -- (n1);
  \draw[->] (part) -- (np);
  \draw[->] (n1) -- (comp);
  \draw[->] (np) -- (comp);
  \draw[->] (comp) -- (frame);
\end{tikzpicture}
\caption{Sort-first distributed rendering with the node-local support guard.
The guard (Algorithm~\ref{alg:guard}) is a per-primitive map applied independently
on each node before it renders its tile, so an out-of-range primitive is clamped at
the node that owns it and cannot be broadcast into the composite. The guard adds no
synchronization point and no cross-node traffic.}
\label{fig:pipeline}
\end{figure}
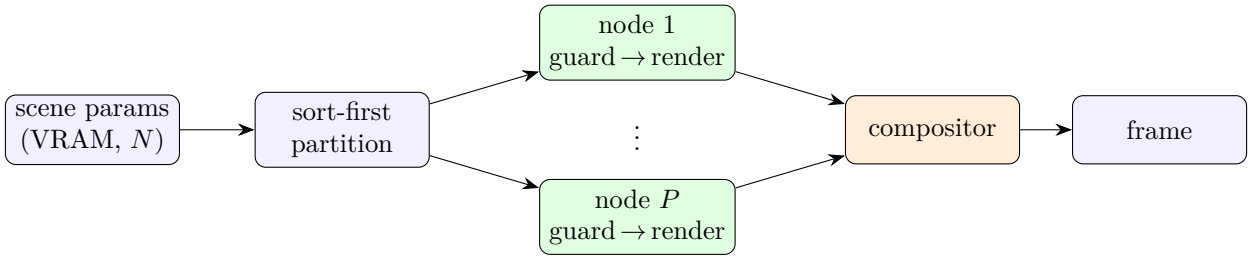

Contamination is a count of nodes, but the operational cost in a
barrier-synchronized compositor is set by the slowest rank. We therefore render
each node region independently, as a sort-first rank would, and time it. With a
sixteen-region partition the slowest-rank time rises from \rankBarrierClean{}~ms on
the clean scene to \rankBarrierCorrupt{}~ms under an unguarded scale-sign upset,
because the exploded primitive must be processed by every rank whose region it
overlaps, and the load imbalance across ranks reaches \rankImbalCorrupt$\times$ the
mean. The node-local guard removes both effects, returning the slowest-rank time to
\rankBarrierGuard{}~ms and the imbalance to \rankImbalGuard$\times$.

We then validate this on real hardware rather than by emulation. Running the
renderer across \mgpuWorld{} physical GPUs as two ranks, each rendering one screen
half at $\mgpuRenderW{}^2$ and contributing it to the composite over the PCIe
interconnect by an NCCL all-gather measured at \mgpuTransferGbps~GB/s, a single
scale-sign upset contaminates both GPU nodes (\mgpuContamNg{} of \mgpuWorld{}),
whereas the node-local guard, applied independently on each GPU's replica before it
renders, confines the corruption to \mgpuContamG{} node. The per-rank render time
stays balanced at \mgpuRankMs~ms and the composited frame takes \mgpuFrameMs~ms.
The same effect holds at higher node counts: on a \mgpuFourWorld{}-GPU L40S system the
unguarded scale-sign upset reaches all \mgpuFourContamNg{} nodes, which the node-local
guard confines to \mgpuFourContamG{}. These are genuine multi-GPU measurements over a
real interconnect; scaling to a datacenter network of many nodes, where the
contamination count sets the cross-node traffic, remains future work.

\subsection{Predicting criticality without injection}
\label{sec:results-predict}

The concentration of risk suggests that criticality is predictable from static
features of a fault site rather than from a render. We train a classifier to
predict whether an upset is catastrophic from only the field, the bit position,
the bit class, and the stored value. It reaches an area under the ROC curve of
\predAUC{}; a classifier given only the field and bit
position, with no value information, reaches \predAUCfieldbit{}, and the single
most informative feature is the \predTopFeat{}. Criticality is therefore almost
entirely determined by which field and which bit are struck, so a future model can
be screened for its vulnerable bits without running an injection campaign. The
prediction also transfers across scenes: trained on all but one scene and tested on
the held-out one, the classifier keeps a leave-one-scene-out area under the ROC
curve of \predAUCcross{} (minimum \predAUCcrossMin{} across the held-out scenes),
which is expected because the criticality is set by the floating-point layout and
the activations rather than by scene content.

\subsection{Throughput and feasibility}
\label{sec:results-throughput}

The engine is a parallel-computing artifact, and its throughput underwrites the
campaign. The rasterizer absorbs additional cameras at near-constant per-camera
cost up to roughly sixteen simultaneous views, reaching \renderPeakMpix{}
megapixels per second before it saturates the device, which is why batching the
$K$ views of each injection into one call is the throughput-critical choice. The
campaign sustained \meanUtil\% utilization for \gpuHours{}. Reproducing it on a
single-threaded CPU rasterizer would require an estimated \cpuDays{}, the concrete
sense in which the device was used to do more work rather than the same work
faster. The per-injection campaign is latency-bound, because each corrupted render
is small relative to the host control between renders. Batching the corruption
itself removes that limit: rendering \batchB{} independently corrupted variants of
the scene in a single rasterizer call, with the bit flips applied to the batch by a
vectorized scatter, saturates the device at \batchUtil\% utilization and
\batchPower~W and sustains \batchInjPerSec{} single-bit injections per second,
\batchGaussInst{} million primitive instances per launch. The fault-injection
engine is therefore not merely a script around a renderer but a throughput-bound
parallel workload in its own right.

The campaign is also embarrassingly parallel across devices, since independent
upsets need no communication. Run data-parallel on \scaleFourNodes{} L40S accelerators
it reaches \scaleFourAgg{} single-bit injections per second, a \scaleFourSpeedup$\times$
speedup over one of them at \scaleFourEff\% parallel efficiency, with all
\scaleFourNodes{} devices held at \scaleFourUtil\% utilization. The bit-level criticality
is unchanged on this hardware: the L40S is an Ada-generation accelerator with ECC
memory while the device used above is Blackwell, yet the field-and-bit ordering is
identical, which confirms that the criticality is a property of the representation
and the floating-point layout rather than of any one architecture. Only the
single-device throughput differs with the hardware, \lFortySingleInj{} injections per
second on an L40S against \batchInjPerSec{} on the Blackwell device, tracking their
memory bandwidth rather than any change in the result.

The synthetic models are small by design, so that a per-bit catastrophe rate can
be estimated with thousands of samples per cell; the guard cost, however, is a
function of model size and should be reported as such rather than at a single
operating point. We therefore replicate the real scene to a range of sizes that
reach \maxStressN{} primitives and saturate the device memory, using up to
\vramMax{}~GB of VRAM, and measure the guard cost and render throughput at each.
The guard cost grows close to linearly with
the primitive count, as expected for a bounded pass over the parameters that the
guard clamps. Because the guard skips the inert higher-order spherical-harmonic
coefficients, which are forty-five of the fifty-nine components per primitive, that
pass touches only a quarter of the parameter memory: it reaches \guardMsBig{}~ms
per frame at \maxStressN{} primitives, \guardFracBig\% of the render at that size,
while the render sustains \mpixBig{} megapixels per second. The guard also need not
run every frame; deployed as a periodic scrub its per-frame cost amortizes toward
zero, and these per-frame numbers are the upper bound. The guard continues to remove the catastrophic tail at this scale, with
the mean scale-sign footprint falling from \bigScaleFootNg\% to \bigScaleFootG\%, on
a model holding \bigParamBits{} billion stored bits, and it sustains an effective
parameter-read bandwidth of \guardBwBig{}~GB/s. We further drive the largest model
with a continuous fault storm of \stormK{} random upsets injected every frame over
\stormFrames{} frames, the closest analogue to a sustained radiation environment at
real-time rates: the guarded frame latency holds at \stormLatG{}~ms against
\stormLatNg{}~ms unguarded. The guard cost is therefore not a fixed small constant
measured in isolation but a quantity that tracks model size, exercises a meaningful
fraction of memory bandwidth, and stays a small part of the render it protects.

\section{Discussion and Limitations}
\label{sec:discussion}

Our fault model is the single-bit upset in the stored checkpoint
representation, which is the form a model takes in VRAM between frames. A
renderer that pre-activates parameters into a separate cache would expose those
cached values as additional targets; we expect the qualitative concentration to
carry over, since the same activations are involved, but the exact bit ordering
would shift, and we mark this as the most useful extension.

The first-order bound of Proposition~\ref{prop:image} holds only while the upset
does not change which pixels a primitive covers, and therefore does not change the
tile binning or the front-to-back sort. This is exactly the small-perturbation
regime of the mantissa bits, which is where we apply it; the catastrophic regime of
the sign and high exponent bits violates local linearity precisely because the
exploded primitive rewrites the sort keys and the tile assignment, and we do not
use the linear bound there. That regime is instead characterized by the exact value
perturbation of Lemma~\ref{lem:value} and by the measured footprint, and the
guard's guarantee in Theorem~\ref{thm:complete} is purely geometric, a bound on the
projected scale, so it does not depend on local linearity at all. The
higher-order spherical-harmonic coefficients are inert for a related reason: they
modulate a primitive's color as a function of view direction but cannot change its
spatial extent, so even an aggressive view-dependent specular flip can only
saturate color inside the footprint the primitive already has, which is bounded by
the scale, rather than expand it. The support guard does not bound color
saturation, but Figure~\ref{fig:heatmap} shows it is a negligible contributor to
the footprint, and a per-field color clamp would cover it at the same cost if
needed.

Our fault model reinterprets the stored value as an integer and flips a bit, which
captures an upset that reaches the rendered output but not the hardware masking and
correction that precede it. Server-class accelerators with ECC memory correct
single-bit and detect double-bit upsets, so on such hardware the relevant residual
is the multi-bit and miscorrected events that ECC misses, while consumer and many
embedded GPUs expose the unprotected case we model directly. The support guard is
complementary to ECC rather than a replacement: it operates on the semantic range
of the parameters, not on memory words, so it catches errors that survive or bypass
ECC, including those introduced after correction in registers or caches. It is
important that the guard not silently mask a degrading device. Because a clamp event
is a detection of an out-of-range value, the guard should count and surface its
clamps as a silent-data-corruption signal to the host, so that a rising clamp rate
triggers the same operating-system parity alerts and structural reboot that a
hardware scrubber would; used this way the guard is a cheap detector as well as a
corrector, and it complements periodic memory scrubbing by bounding the impact of
any upset that occurs within a scrub interval.

The guard depends on the trained support box, so its behavior is tied to training
dynamics. Densification and opacity resetting drive the scale distribution toward
small, compact primitives, which is what makes the box tight and the clamp
effective: the largest trained log scale sets the bound in Theorem~\ref{thm:complete},
and a tighter distribution yields a smaller worst-case footprint. Aggressive
optimization that deliberately grows a few very large primitives, for example to
cover a background or a sky with one splat, would widen the box and admit a larger
guarded footprint, which is the one regime in which the primitive-locality
assumption weakens. The box is also computed per component over the converged
model, so it does not protect against a flip that lands within the box yet is wrong
for that primitive; Lemma~\ref{lem:value} bounds that residual, and the predictor
of Section~\ref{sec:results-predict} could supply a per-primitive bound in its
place when a model contains atypically large primitives. A practical deployment
would recompute the box whenever the model is retrained or edited, which is cheap
since it is a single pass over the parameters.

Dynamic and deformable Gaussian splatting, where primitives move and change over
time, is the setting in which the box itself becomes time varying. The guard
extends naturally because its only state is the per-field box and that box is a
single pass to recompute. For a deforming scene the box can be refreshed per
keyframe, or a temporal margin can be carried so that the bound tracks the moving
support without a per-frame pass; the geometric guarantee of
Theorem~\ref{thm:complete} then holds against the current box, and a transient
fault that throws a parameter outside the deformation envelope is still caught. The
one new failure mode is a fault that corrupts the deformation field or the temporal
state rather than a static parameter, which our model does not cover and which we
flag as the natural extension for dynamic representations. The distributed
analysis is conducted by emulation on one device rather than on a physical
cluster, so it speaks to the contamination geometry, the relationship between
footprint and tile count, rather than to network or scheduling effects. We study
trained synthetic scenes at a fixed resolution; larger real-world scenes have
more primitives and therefore even more redundancy, which would deepen rather
than reverse the concentration finding. We inject faults numerically; we
do not claim absolute SEU rates, which are a property of a given device and
environment and are the subject of the beam-testing literature
\cite{oliveira2016evaluation}. Our throughput and bandwidth figures are wall-clock
and counter-free: instruction-level metrics such as instructions per cycle require
hardware performance counters that are not exposed inside the unprivileged
container we ran in, so we report achieved pixel throughput, effective memory
bandwidth, and frame latency under load instead. The contribution is the mapping
from a single-bit upset to its effect on the rendered image, and the defense
derived from it.

\section{Conclusion}
\label{sec:conclusion}

This paper characterized the effect of single-event upsets on 3D Gaussian
splatting rendering and derived a defense from that characterization. The effect
is concentrated: because the representation is highly redundant, most single-bit
upsets are perceptually invisible, and the exceptions are predictable from the
floating-point layout composed with the rendering activations, dominated by the
sign bit of the logarithmic scale, which enlarges a primitive across the frame. A
closed-form bound reproduces the per-bit ordering. The support guard, a
per-primitive clamp to the trained parameter box, costs a fraction of one rendered
frame, parallelizes across primitives, leaves a clean model unchanged, and is
proved to make frame-covering corruption impossible under any single-bit upset; it
removed the catastrophic events across the full campaign and kept the rendered
quality close to clean under accumulated dose where the unprotected renderer
degrades. With this guard, a representation designed for rendering speed also
tolerates single-event upsets, the property required for deployment on spaceborne,
robotic, and clustered hardware.

\FloatBarrier
\bibliographystyle{plain}
\bibliography{refs}

@article{kerbl2023gaussian,
  title={3D Gaussian Splatting for Real-Time Radiance Field Rendering},
  author={Kerbl, Bernhard and Kopanas, Georgios and Leimk{\"u}hler, Thomas and Drettakis, George},
  journal={ACM Transactions on Graphics},
  volume={42},
  number={4},
  year={2023},
  publisher={ACM}
}

@inproceedings{mildenhall2020nerf,
  title={{NeRF}: Representing Scenes as Neural Radiance Fields for View Synthesis},
  author={Mildenhall, Ben and Srinivasan, Pratul P and Tancik, Matthew and Barron, Jonathan T and Ramamoorthi, Ravi and Ng, Ren},
  booktitle={European Conference on Computer Vision (ECCV)},
  year={2020}
}

@article{ye2024gsplat,
  title={gsplat: An Open-Source Library for {Gaussian} Splatting},
  author={Ye, Vickie and Li, Ruilong and Kerr, Justin and Turkulainen, Matias and Yi, Brent and Pan, Zhuoyang and Seiskari, Otto and Ye, Jianbo and Hu, Jeffrey and Tancik, Matthew and Kanazawa, Angjoo},
  journal={arXiv preprint arXiv:2409.06765},
  year={2024}
}

@inproceedings{matsuki2024gaussian,
  title={Gaussian Splatting {SLAM}},
  author={Matsuki, Hidenobu and Murai, Riku and Kelly, Paul H J and Davison, Andrew J},
  booktitle={IEEE/CVF Conference on Computer Vision and Pattern Recognition (CVPR)},
  year={2024}
}

@article{dixit2021silent,
  title={Silent Data Corruptions at Scale},
  author={Dixit, Harish Dattatraya and Pendharkar, Sneha and Beadon, Matt and Mason, Chris and Chakravarthy, Tejasvi and Muthiah, Bharath and Sankar, Sriram},
  journal={arXiv preprint arXiv:2102.11245},
  year={2021}
}

@inproceedings{hochschild2021cores,
  title={Cores that don't count},
  author={Hochschild, Peter H and Turner, Paul and Mogul, Jeffrey C and Govindaraju, Rama and Ranganathan, Parthasarathy and Culler, David E and Vahdat, Amin},
  booktitle={Workshop on Hot Topics in Operating Systems (HotOS)},
  year={2021}
}

@article{oliveira2016evaluation,
  title={Evaluation and Mitigation of Radiation-Induced Soft Errors in Graphics Processing Units},
  author={Oliveira, Daniel A G de and Pilla, Laercio L and Santini, Thiago and Rech, Paolo},
  journal={IEEE Transactions on Computers},
  volume={65},
  number={3},
  pages={791--804},
  year={2016}
}

@inproceedings{fratin2018code,
  title={Code-Dependent and Architecture-Dependent Reliability Behaviors},
  author={Fratin, Vinicius and Oliveira, Daniel and Lunardi, Caio and Santos, Fernando and Rodrigues, Gennaro and Rech, Paolo},
  booktitle={IEEE/IFIP International Conference on Dependable Systems and Networks (DSN)},
  year={2018}
}

@article{baumann2005radiation,
  title={Radiation-Induced Soft Errors in Advanced Semiconductor Technologies},
  author={Baumann, Robert C},
  journal={IEEE Transactions on Device and Materials Reliability},
  volume={5},
  number={3},
  pages={305--316},
  year={2005}
}

@article{normand1996single,
  title={Single Event Upset at Ground Level},
  author={Normand, Eugene},
  journal={IEEE Transactions on Nuclear Science},
  volume={43},
  number={6},
  pages={2742--2750},
  year={1996}
}

@book{mukherjee2008architecture,
  title={Architecture Design for Soft Errors},
  author={Mukherjee, Shubu},
  publisher={Morgan Kaufmann},
  year={2008}
}

@article{huang1984algorithm,
  title={Algorithm-Based Fault Tolerance for Matrix Operations},
  author={Huang, Kuang-Hua and Abraham, Jacob A},
  journal={IEEE Transactions on Computers},
  volume={C-33},
  number={6},
  pages={518--528},
  year={1984}
}

@inproceedings{liu2017fault,
  title={Fault Injection Attack on Deep Neural Network},
  author={Liu, Yannan and Wei, Lingxiao and Luo, Bo and Xu, Qiang},
  booktitle={IEEE/ACM International Conference on Computer-Aided Design (ICCAD)},
  year={2017}
}

@inproceedings{rakin2019bitflip,
  title={Bit-Flip Attack: Crushing Neural Network with Progressive Bit Search},
  author={Rakin, Adnan Siraj and He, Zhezhi and Fan, Deliang},
  booktitle={IEEE/CVF International Conference on Computer Vision (ICCV)},
  year={2019}
}

@inproceedings{mahmoud2020pytorchfi,
  title={{PyTorchFI}: A Runtime Perturbation Tool for {DNNs}},
  author={Mahmoud, Abdulrahman and Aggarwal, Neeraj and Nobbe, Alex and Vicarte, Jose Rodrigo Sanchez and Adve, Sarita V and Fletcher, Christopher W and Frosio, Iuri and Hari, Siva Kumar Sastry},
  booktitle={IEEE/IFIP Int. Conf. on Dependable Systems and Networks Workshops (DSN-W)},
  year={2020}
}

@inproceedings{li2017understanding,
  title={Understanding Error Propagation in Deep Learning Neural Network ({DNN}) Accelerators and Applications},
  author={Li, Guanpeng and Hari, Siva Kumar Sastry and Sullivan, Michael and Tsai, Timothy and Pattabiraman, Karthik and Emer, Joel and Keckler, Stephen W},
  booktitle={Int. Conference for High Performance Computing, Networking, Storage and Analysis (SC)},
  year={2017}
}

@inproceedings{reagen2018ares,
  title={Ares: A Framework for Quantifying the Resilience of Deep Neural Networks},
  author={Reagen, Brandon and Gupta, Udit and Pentecost, Lillian and Whatmough, Paul and Lee, Sae Kyu and Mulholland, Niamh and Brooks, David and Wei, Gu-Yeon},
  booktitle={Design Automation Conference (DAC)},
  year={2018}
}

@article{geissler2021range,
  title={Towards a Safety Case for Hardware Fault Tolerance in Convolutional Neural Networks Using Activation Range Supervision},
  author={Geissler, Florian and Qutub, Syed and Roychowdhury, Sumanta and Asgari, Ali and Peng, Yang and Dhamasia, Akash and Graefe, Ralf and Pattabiraman, Karthik and Paulitsch, Michael},
  journal={arXiv preprint arXiv:2108.07019},
  year={2021}
}

@inproceedings{hari2017sassifi,
  title={{SASSIFI}: An Architecture-Level Fault Injection Tool for {GPU} Application Resilience Evaluation},
  author={Hari, Siva Kumar Sastry and Tsai, Timothy and Stephenson, Mark and Keckler, Stephen W and Emer, Joel},
  booktitle={IEEE International Symposium on Performance Analysis of Systems and Software (ISPASS)},
  year={2017}
}

@inproceedings{tsai2021nvbitfi,
  title={{NVBitFI}: Dynamic Fault Injection for {GPUs}},
  author={Tsai, Timothy and Hari, Siva Kumar Sastry and Sullivan, Michael and Li, Xinghua and Keckler, Stephen W},
  booktitle={IEEE/IFIP International Conference on Dependable Systems and Networks (DSN)},
  year={2021}
}

@article{ieee754,
  title={{IEEE} Standard for Floating-Point Arithmetic},
  author={{IEEE}},
  journal={IEEE Std 754-2019},
  year={2019}
}

@inproceedings{zhang2018unreasonable,
  title={The Unreasonable Effectiveness of Deep Features as a Perceptual Metric},
  author={Zhang, Richard and Isola, Phillip and Efros, Alexei A and Shechtman, Eli and Wang, Oliver},
  booktitle={IEEE/CVF Conference on Computer Vision and Pattern Recognition (CVPR)},
  year={2018}
}

@article{wang2004image,
  title={Image Quality Assessment: From Error Visibility to Structural Similarity},
  author={Wang, Zhou and Bovik, Alan C and Sheikh, Hamid R and Simoncelli, Eero P},
  journal={IEEE Transactions on Image Processing},
  volume={13},
  number={4},
  pages={600--612},
  year={2004}
}

@article{molnar1994sorting,
  title={A Sorting Classification of Parallel Rendering},
  author={Molnar, Steven and Cox, Michael and Ellsworth, David and Fuchs, Henry},
  journal={IEEE Computer Graphics and Applications},
  volume={14},
  number={4},
  pages={23--32},
  year={1994}
}

\appendix
\section{Proofs}
\label{app:proofs}

\begin{proof}[Proof of Lemma~\ref{lem:value}]
Write a normal value as $\theta=(-1)^{\epsilon}2^{e}(1+m)$ with
$m=\sum_{k=1}^{p} m_k 2^{-k}$, $m_k\in\{0,1\}$, and exponent field
$E=e+\mathrm{bias}$ stored in the exponent bits. The mantissa bit at position
$b$ (counting from the least significant, $b=0,\dots,p-1$) has weight
$2^{-(p-b)}$ in $1+m$, hence weight $2^{e}\cdot 2^{-(p-b)}=2^{e+b-p}$ in $\theta$.
Flipping it adds or removes exactly that weight, so $|\Delta\theta|=2^{e+b-p}$.
Flipping the sign bit sends $\theta\mapsto-\theta$, so $\Delta\theta=-2\theta$.
Flipping exponent bit $j$ (of weight $2^{j}$ in $E$) sends
$E\mapsto E\pm2^{j}$, hence $\theta\mapsto\theta\cdot 2^{\pm2^{j}}$; for the
most significant exponent bits $2^{j}$ exceeds the remaining exponent range and
the result is $\pm\infty$ or a subnormal underflow. The identical argument with
$p=10$ and $p=7$ gives the \texttt{fp16} and \texttt{bf16} cases.
\end{proof}

\begin{proof}[Proof of Proposition~\ref{prop:image}]
$I=R(\phi(\theta),\cdot)$ is differentiable in $\theta$ wherever the tile
assignment and the front-to-back ordering are fixed, which holds for
perturbations that do not change the set of pixels a primitive covers. A
first-order Taylor expansion gives
$\Delta I=\frac{\partial R}{\partial\phi}\phi'(\theta)\Delta\theta+O(\Delta\theta^2)$,
and taking the $\infty$-norm yields the stated bound. For $\phi=\exp$ we have
$\phi'(\theta)=e^{\theta}=s$, so $|\Delta s|=s\,|\Delta\theta|$ and the relative
change $|\Delta s|/s=|\Delta\theta|$ is independent of the exponent of $\theta$;
substituting the mantissa result of Lemma~\ref{lem:value} gives
$|\Delta s|/s=2^{b-p}$. A sign flip sends a trained log scale
$\theta\approx-3$ to $+3$, so $s\mapsto s\cdot e^{6}$.
\end{proof}

\begin{proof}[Proof of Theorem~\ref{thm:guard}]
(i) For a clean model every component lies in its own observed box $B_f$ by
definition of $B_f$ as the per-component min--max over that model, so the clamp
is the identity and no rendered pixel changes. (ii) After a single-bit upset only
one component $x$ of one primitive changes. If the corrupted value $x'$ lies
outside $B_f$, the guard maps it to the nearer endpoint of $B_f$, so the guarded
value lies in $B_f$ and differs from the clean value $x\in B_f$ by at most
$h_f-\ell_f$, the in-box spread; in particular every upset whose corrupted value
left $B_f$, which includes all sign flips that change magnitude order, all
overflowing exponent flips, and all non-finite results, is mapped back into the
trained support and cannot explode a primitive. If instead $x'\in B_f$, the guard
leaves it unchanged and the residual equals the original in-support perturbation,
which by Lemma~\ref{lem:value} is at most $2^{e+b-p}$ for a mantissa flip.
\end{proof}

\begin{proof}[Proof of Theorem~\ref{thm:complete}]
By construction the guard maps every stored component into its box $B_f$ before
rendering, so after guarding all parameters of every primitive lie in their boxes
regardless of how many bits were flipped. The projected screen extent of a
primitive is a monotone function of its scale $s=\exp(\theta_{\mathrm{scale}})$,
and $\theta_{\mathrm{scale}}\in B_{\mathrm{scale}}=[\ell,h]$ forces
$s\le\exp(h)=s_{\max}$, the largest scale present in the clean trained model.
Hence every primitive's footprint is bounded by that of the largest trained
primitive, which for a converged scene is a small fraction of the frame. No
choice of flipped bit can exceed this bound, since the bound depends only on
$B_{\mathrm{scale}}$ and not on the corrupted value. The single failure mode that
the redundancy of the representation cannot absorb is a primitive whose footprint
spans the frame, and that is now excluded.
\end{proof}

\begin{proof}[Proof of Corollary~\ref{cor:adv}]
A flip of any number of bits within a single primitive still yields, after the
guard, a primitive with all parameters in $B_f$, whose footprint is bounded by
that of the largest trained primitive (Theorem~\ref{thm:complete}); call this
bound $\phi$ as a fraction of the frame. An adversary spending $m$ flips touches
at most $m$ distinct primitives, and the union of their footprints is at most
$m\phi$ of the frame. To exceed any target fraction $\beta$ the adversary needs
$m\ge\beta/\phi$ flips, in contrast to the single flip that produces a
near-full-frame footprint without the guard.
\end{proof}

\begin{proof}[Proof of Lemma~\ref{lem:redundancy}]
Fix the scene as a bounded set of surfaces rendered to a frame of $P$ pixels, and
let primitive $i$ have projected footprint $F_i$, the set of pixels at which it
changes the rendered value by more than the $1/255$ threshold; write
$f_i=|F_i|/P\in[0,1]$ for its footprint fraction. For a converged model the
per-pixel overdraw is bounded: each pixel accumulates non-negligible contributions
from at most a constant number $\bar C$ of primitives before the front-to-back
transmittance saturates, so
\[
\sum_{i=1}^{N} f_i \;=\; \frac{1}{P}\sum_{\text{pixels }x}\#\{i:x\in F_i\}\;\le\;\bar C,
\]
a bound independent of $N$ (the \emph{overdraw bound}).

\emph{Per-upset error of a non-catastrophic upset.} A single-bit upset alters one
component of one primitive $i$. Call the upset \emph{non-catastrophic} if the
corrupted value remains inside the trained support box $B_f$ (equivalently, the
guard of Theorem~\ref{thm:complete} would not clamp it); let $G$ denote this event,
with $q=\Pr[G^{c}]$ under a uniform-random choice of primitive, field, and bit. By
the criticality measurement (Table~\ref{tab:criticality}), $q$ is a strictly
positive constant fixed by the fraction of magnitude-order sign and high-exponent
bits and is independent of $N$. On $G$, Proposition~\ref{prop:image} bounds the
perturbation pixelwise by a constant $\Delta_{\max}$ and, by the locality of
splatting with all other primitives unchanged, confines it to $F_i$. Hence the
frame-normalized squared error of an upset that strikes primitive $i$ on $G$
satisfies
\[
\varepsilon_i \;=\; \frac{1}{P}\sum_{x\in F_i}\lvert\Delta I(x)\rvert^{2}\;\le\; f_i\,\Delta_{\max}^{2},
\qquad\text{and}\qquad \varepsilon_i=\Theta(f_i)
\]
for the upsets whose per-pixel change is of the order of the bound. Up to constants
the per-upset error is therefore the footprint fraction $f_i$ of the struck
primitive, and a uniform-random upset selects $i$ uniformly over the $N$ primitives.

\emph{Median scaling.} Let $f_{(1)}\le\cdots\le f_{(N)}$ be the ordered footprint
fractions and $f_{\mathrm{med}}(N)=f_{(\lceil N/2\rceil)}$ their median. The
overdraw bound gives mean footprint $\frac1N\sum_i f_i\le \bar C/N$, so by Markov's
inequality at least half the primitives satisfy $f_i\le 2\bar C/N$, whence
$f_{\mathrm{med}}(N)\le 2\bar C/N=O(N^{-1})$. For the matching lower order,
densification refines the tiling: each new primitive's footprint is contained in
that of the parent it splits, so increasing the count from $N$ to $N'$ partitions
the same footprint mass among more primitives and the ordered footprints obey
$f_{(\lceil N/2\rceil)}=\Theta(N^{-\alpha})$ for an exponent $\alpha\in(0,1]$ set by
the size distribution of the tiling and the projection: $\alpha=1$ for a regular
refinement in which every primitive is split, and $\alpha<1$ when refinement
concentrates in part of the scene, so the median footprint shrinks more slowly than
the mean. In every case $\alpha>0$, because the largest footprint $f_{(N)}\to0$ as
the fixed surface area is partitioned among unboundedly many primitives. Combining
with the per-upset bound, the median error over non-catastrophic upsets is
$\sigma^{2}(N)=\Theta\!\big(f_{\mathrm{med}}(N)\big)=\Theta(N^{-\alpha})$; since $G$
has constant probability $1-q$, the same scaling holds for the median over all
upsets, which is the first claim. The exponent $\alpha$ is scene-dependent and is
measured in Section~\ref{sec:results-accum}.

\emph{The mean does not scale.} Write the mean frame-normalized squared error over
all upsets as
\[
\mathbb{E}[\varepsilon]=(1-q)\,\mathbb{E}[\varepsilon\mid G]+q\,\mathbb{E}[\varepsilon\mid G^{c}].
\]
The first term is at most $(1-q)\,\bar C\,\Delta_{\max}^{2}/N=O(N^{-1})$ by the
overdraw and per-upset bounds. For the second, a catastrophic upset is a
magnitude-order sign flip or an overflowing exponent flip of the logarithmic scale;
absent the guard it inflates one primitive to a frame-spanning sheet, so its
footprint fraction is $\Theta(1)$ \emph{independently of $N$}---exactly the failure
mode excluded by Theorem~\ref{thm:complete}---whence
$\mathbb{E}[\varepsilon\mid G^{c}]=\Theta(1)$. As $q$ is a positive constant,
\[
\mathbb{E}[\varepsilon]\;=\;O(N^{-1})+q\,\Theta(1)\;=\;\Theta(1),
\]
which is $N$-independent: the catastrophic tail pins the mean while the median
decays as $N^{-\alpha}$. Removing $G^{c}$ with the guard
(Theorem~\ref{thm:complete}) restores $\mathbb{E}[\varepsilon\mid G]=\Theta(N^{-\alpha})$,
the additive-law regime invoked in Theorem~\ref{thm:dose}.
\end{proof}

\begin{proof}[Proof of Theorem~\ref{thm:dose}]
Let $\Delta_i$ be the image perturbation of the $i$-th upset. The upsets are drawn
independently and hit distinct primitives with probability $1-O(k/N)$, so their
changed-pixel supports are disjoint in expectation and
$\mathbb{E}\langle\Delta_i,\Delta_j\rangle=0$ for $i\ne j$. Then
$\mathbb{E}\lVert\sum_i\Delta_i\rVert^2=\sum_i\mathbb{E}\lVert\Delta_i\rVert^2=k\,\sigma^2(N)$.
Setting $k\,\sigma^2(N)=\tau$ with $\sigma^2(N)=\Theta(N^{-\alpha})$ gives
$k_\tau=\Theta(N^{\alpha}\tau)$. Catastrophic upsets violate the disjoint-support
assumption, since one explosion covers the frame; the guard removes them by
Theorem~\ref{thm:complete}, restoring the additive law up to the in-support
residual.
\end{proof}

\begin{proof}[Proof of Corollary~\ref{cor:scrub}]
Bits flip independently at rate $\lambda$ over $b=\Theta(N)$ stored bits, so the
expected number of live upsets $M$ frames after a scrub is $\lambda b M$. Requiring
$\lambda b M\le k_\tau$ and substituting $k_\tau=\Theta(N^{\alpha}\tau)$ with
$b=\Theta(N)$ gives $M\le\Theta(N^{\alpha-1}\tau/\lambda)$. The guard is a
per-primitive map (Algorithm~\ref{alg:guard}) with no cross-node dependence, so on
$P$ nodes it costs $O(N/P)$ per invocation, and invoked once per $M$ frames its
amortised per-frame cost is $O(N/(PM))$.
\end{proof}

\section{Reproducibility}
\label{app:repro}

The ancillary archive contains the training script, the fault-injection engine,
the campaign and analysis code, the aggregated per-cell records, the multi-upset
records, the throughput measurements, the raw logs including the periodic
device-utilization trace, and the scripts that regenerate every figure and table
in this paper. The trained models and the several-million-row raw per-injection
records are large and are not shipped in the archive, but the code regenerates
them deterministically: training uses \texttt{gsplat} densification on the four
synthetic scenes and completes in minutes per scene, and the campaign parameters
(samples per cell, views per injection, precisions, and random seeds) are fixed
in the scripts and recorded in the logs. The complete artifact set, including the
trained models and the regenerating code, is hosted at
\url{https://huggingface.co/datasets/Lightcap/seu-3dgs}.

\end{document}